\documentclass[onecolumn, 10pt]{IEEEtran}
\usepackage{amsmath,amssymb,amsthm}
\usepackage{bm}
\usepackage{url}
\usepackage{tikz}

\newtheorem{theorem}{Theorem}
\newtheorem{lemma}[theorem]{Lemma}

\theoremstyle{definition}
\newtheorem{definition}[theorem]{Definition}
\newtheorem{example}[theorem]{Example}
\theoremstyle{remark}
\newtheorem{remark}{Remark}

\title{Linear Programming Relaxations for Goldreich's Generators over Non-Binary Alphabets}
\author{Ryuhei~Mori, Takeshi Koshiba, Osamu Watanabe, and Masaki Yamamoto%
\thanks{This work was supported by MEXT KAKENHI Grant Number 24106008.}%
\thanks{R. Mori is with the Department of Mathematical and Computing Science, Graduate School of Information Science and Engineering,
Tokyo Institute of Technology, Shibaura, Minato-ku, Tokyo, 108-0023 Japan
(e-mail: mori@is.titech.ac.jp).}%
\thanks{T. Koshiba is with the Division of Mathematics, Electronics and Informatics,
Graduate School of Science and Engineering, Saitama University,
255 Shimo-Okubo, Sakura-ku, Saitama, 338-8570 Japan 
(e-mail: koshiba@mail.saitama-u.ac.jp)
}%
\thanks{O. Watanabe is with the Department of Mathematical and
Computing Science, Graduate School of Information Science and
Engineering, Tokyo Institute of Technology, Ookayama, Meguro-ku,
Tokyo, 152-0038 Japan (e-mail: watanabe@is.titech.ac.jp).
}%
\thanks{M. Yamamoto is with the Department of Computer and Information Science, Seikei University,
Musashino-shi, Tokyo, 180-8633 Japan (e-mail: yamamoto@st.seikei.ac.jp).
}%
}

\begin{document}
\maketitle
\begin{abstract}
Goldreich suggested candidates of one-way functions and pseudorandom generators included in $\mathsf{NC}^0$.
It is known that randomly generated Goldreich's generator using $(r-1)$-wise independent predicates with $n$ input variables
and $m=C n^{r/2}$ output variables is not pseudorandom generator with high probability for sufficiently large constant $C$.
Most of the previous works assume that the alphabet is binary and use techniques available only for the binary alphabet.
In this paper, we deal with non-binary generalization of Goldreich's generator and
derives the tight threshold for linear programming relaxation attack using local marginal polytope for randomly generated Goldreich's generators.
We assume that $u(n)\in \omega(1)\cap o(n)$ input variables are known.
In that case, we show that when $r\ge 3$, there is an exact threshold
$\mu_\mathrm{c}(k,r):=\binom{k}{r}^{-1}\frac{(r-2)^{r-2}}{r(r-1)^{r-1}}$ such that
for $m=\mu\frac{n^{r-1}}{u(n)^{r-2}}$,
the LP relaxation can determine linearly many input variables of Goldreich's generator if $\mu>\mu_\mathrm{c}(k,r)$, and that
the LP relaxation cannot determine $\frac1{r-2} u(n)$ input variables of Goldreich's generator if $\mu<\mu_\mathrm{c}(k,r)$.
This paper uses characterization of LP solutions by combinatorial structures called stopping sets on a bipartite graph,
which is related to a simple algorithm called peeling algorithm.
\end{abstract}

\section{Introduction}
Goldreich suggested candidates of one-way functions (OWFs) using predicates for constant number of binary variables
 and an expander bipartite graph~\cite{goldreich2000candidate}.
It is conjectured that Goldreich's idea also gives pseudorandom generators (PRGs)~\cite{cryan2001pseudorandom}, \cite{mossel2006varepsilon}.
Since every bit in the output of Goldreich's generators only depends on a constant number $k$ of input bits,
Goldreich's generators are in $\mathsf{NC}^0$, which means that the generators are extremely simple.
Applebaum, Ishai, Kushilevitz showed that if OWF (PRG) exists in $\mathsf{NC}^1$ then OWF (PRG) exists also in $\mathsf{NC}^0$, respectively~\cite{doi:10.1137/S0097539705446950}.
That means that OWF and PRG exist in $\mathsf{NC}^0$ on mild assumptions, e.g., hardness of factoring.
Mossel, Shpilka and Trevisan showed that for large $k$, there is a polynomially stretching Goldreich's generator being a small bias generator, which is
a weak pseudorandom generator only fooling linear tests~\cite{mossel2006varepsilon}.
In contrast to the case of small bias generators, it is difficult to show that there exists an one-way function or a pseudorandom generator even in general.

For analyzing Goldreich's generators, many papers investigate properties of randomly generated planted constraint satisfaction problems (CSPs)
since randomly generated bipartite graph is an expander with high probability.
Two types of attacks for randomly generated Goldreich's generators have been known.
The first one is an algebraic attack called ``correlation attack''~\cite{mossel2006varepsilon}.
For given predicate $P$, if $P$ can be expressed as degree-$d$ polynomial on $\mathbb{F}_2$,
then the correlation attack distinguishes
output of Goldreich's generator using $P$ from uniform random variables
if $m>\sum_{i=1}^d\binom{n}{i}$ where $n$ and $m$ are input length and output length of generator, respectively.
Since the correlation attack is a linear test, Goldreich's generator is not small bias generator if $m>\frac1{d!}n^d$.
The second type of attack is based on reduction to planted noisy MAX $r$-LIN problem~\cite{v003a002}.
Here, the concept of $(r-1)$-wise independence of predicate $P$, which will be defined later in this paper,
 gives the critical order of the output length which separates the secure region and the insecure region for the second type of attack
while the degree $d$ of polynomial representation of $P$ on $\mathbb{F}_2$ gives the critical order for the correlation attack.
For given $(r-1)$-wise independent predicate $P$,
Goldreich's generator using the predicate $P$ is insecure as small bias generator for $m> Cn^{r/2}$ for sufficiently large constant $C$
since when $m>Cn^{r/2}$ there is a pair of correlated output variables with high probability from the birthday paradox~\cite{mossel2006varepsilon}.
There also exists attack to Goldreichs' generator using planted noisy MAX $r$-LIN problem
as OWF when $m=Cn^{r/2}\sqrt{\log n}$ for sufficiently large $C$~\cite{v003a002}, \cite{odonell2014goldreich}.
No attack is known for Goldreich's PRG for $m=o(n^{\min\{d, r/2\}})$.
Since the maximum of $\min\{d, r/2\}$ among all predicates for $k$ variables is $\frac12\lfloor\frac23 k\rfloor$, 
it is conjectured that the optimum stretch by Goldreich's PRG with the locality $k$ is $m=o(n^{\frac12\lfloor\frac23 k\rfloor})$~\cite{odonell2014goldreich}.
In~\cite{odonell2014goldreich}, it is shown that the semidefinite programming (SDP) relaxation using Sherali-Adams${}^+$ hierarchy
cannot distinguish output of Goldreich's generator with small modification
 from uniform random variables if $m=O(n^{r/2-\delta})$ for any $\delta>0$.
The above results uses techniques available only for the binary alphabet.
It is not obvious that the above results can be generalized to non-binary alphabets.

In this work, we deal with a generalization of Goldreich's generator to non-binary alphabet and local functions with multiple output variables.
We assume that $u(n)=\omega(1)\cap o(n)$ of $n$ input variables for Goldreich's generator are known,
and derive an exact threshold for the number $m$ of local functions in the generator
 on linear programming (LP) relaxation attack using a simple polytope called local marginal polytope.
On the local marginal polytope, we show that when $r\ge 3$, there is an exact threshold 
$\mu_\mathrm{c}(k,r):=\binom{k}{r}^{-1}\frac{(r-2)^{r-2}}{r(r-1)^{r-1}}$ such that
for $m=\mu\frac{n^{r-1}}{u(n)^{r-2}}$,
the LP relaxation can determine linearly many input variables of Goldreich's generator if $\mu>\mu_\mathrm{c}(k,r)$, and that
the LP relaxation cannot determine $\frac1{r-2} u(n)$ input variables of Goldreich's generator if $\mu<\mu_\mathrm{c}(k,r)$.
This paper uses characterization of LP solutions by combinatorial structures called stopping sets on a bipartite graph,
which is related to a simple algorithm called peeling algorithm.
Since peeling algorithm naturally appears in many problems~\cite{mitzenmacher2012peeling}, our results may have applications also in other areas.

\section{Pseudorandom generators, Goldreich's generators}
\subsection{One-way function and pseudorandom generator}
\begin{definition}[One-way function]
Let $\mathcal{X}$ be a finite alphabet.
For $n\in\mathbb{N}$ and $m\in\mathbb{N}$,
$g\colon \mathcal{X}^n\to\mathcal{X}^m$ is said to be a $\epsilon$-secure one-way function if
\begin{equation*}
\Pr\left(h(g(X))\in g^{-1}(g(X))\right)<\epsilon
\end{equation*}
for any function $h\colon\mathcal{X}^m\to\mathcal{X}^n$ which has a probabilistic polynomial-time algorithm
where $X$ denotes a uniformly distributed random variable on $\mathcal{X}^n$. 
\end{definition}

\begin{definition}[Pseudorandom generator]
Let $\mathcal{X}$ be a finite alphabet.
For $n\in\mathbb{N}$ and $m\in\mathbb{N}$,
$g\colon \mathcal{X}^n\to\mathcal{X}^m$ is said to be a $\epsilon$-secure pseudorandom generator if
\begin{equation*}
\left|\Pr\left(h(g(X))=1\right) - \Pr\left(h(U)=1\right)\right| < \epsilon
\end{equation*}
for any function $h\colon\mathcal{X}^m\to\{0,1\}$ which has a probabilistic polynomial-time algorithm
where $X$ and $U$ denote a uniformly distributed random variable on $\mathcal{X}^n$ and $\mathcal{X}^m$, respectively.
\end{definition}
If OWF and PRG are $1/p(n)$-secure for any polynomial $p(n)$, they are said to be strongly OWF and strongly PRG, respectively.

\subsection{Goldreich's generator and its generalization}\label{subsec:goldreich}
Existence of one-way function and pseudorandom generator is one of the biggest conjecture in computer science and theory of cryptography.
Goldreich suggested extremely simple candidates of one-way function on the binary alphabet in~\cite{goldreich2000candidate}.
Let $k\in\mathbb{N}$ be the locality of the generator.
Then, predicates $P_a\colon \{0,1\}^k\to \{0,1\}$ are fixed for every $a\in\{1,2,\dotsc,m\}$.
For each $a\in\{1,2,\dotsc,m\}$, a $k$-tuple $(i^{(a)}_1,\dotsc,i^{(a)}_k)$ not including duplication is chosen from $\{1,2,\dotsc,n\}$.
For $a\in\{1,2,\dotsc,m\}$, $a$-th bit of output of Goldreich's generator is defined as $P_a(x_{i^{(a)}_1},\dotsc,x_{i^{(a)}_k})$
where $x_i$ denotes $i$-th input bit for $i\in \{1,2,\dotsc,n\}$.
An example of Goldreich's generator is described by a bipartite graph in Fig.~\ref{fig:goldreich}.
Goldreich conjectured that for almost all predicates, the generator is a one-way function if the bipartite graph is expander when $m=n$.
When $m>n$, Goldreich's generator is also candidate of pseudorandom generator~\cite{mossel2006varepsilon}.

In this paper, we consider a generalization of Goldreich's generator to non-binary alphabet and also to local functions with multiple output variables.
Fix the locality $k\in\mathbb{N}$ and the output length $l\in\{1,2,\dotsc,k-1\}$ of the local functions.
Let $q\in\{2,3,4,\dotsc\}$ be the size of alphabet and $[q]:=\{1,2,\dotsc,q\}$ be the alphabet set.
Then, surjective functions $f_a\colon [q]^k \to [q]^{l}$ are fixed for every $a\in\{1,2,\dotsc,m\}$.
For each $a\in\{1,2,\dotsc,m\}$, a $k$-tuple $(i^{(a)}_1,\dotsc,i^{(a)}_k)$ not including duplication is chosen from $\{1,2,\dotsc,n\}$.
Similarly to the original Goldreich's generator,
the output of generalized Goldreich's generator is defined as
$(f_1(x_{i^{(1)}_1},\dotsc,x_{i^{(1)}_k}),\dotsc,f_m(x_{i^{(m)}_1},\dotsc, x_{i^{(m)}_k}))\in[q]^{lm}$.

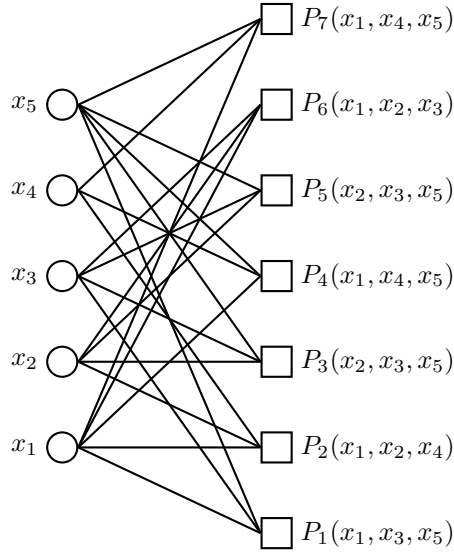
\begin{figure}
\begin{center}
\begin{tikzpicture}
[yshift=20pt, rotate=90,scale=0.57, inner sep=0mm, C/.style={minimum size=4mm,circle,draw=black,thick},
S/.style={minimum size=4mm,rectangle,draw=black,thick}, label distance=1mm]
\node (1) at (0,5.0) [C,label=left:$x_1$] {};
\node (2) at (2,5.0) [C,label=left:$x_2$] {};
\node (3) at (4,5.0) [C,label=left:$x_3$] {};
\node (4) at (6,5.0) [C,label=left:$x_4$] {};
\node (5) at (8,5.0) [C,label=left:$x_5$] {};
\node (a1) at (-2,0) [S,label=right:${P_1(x_1,x_3,x_5)}$] {};
\node (a2) at (0,0) [S,label=right:${P_2(x_1,x_2,x_4)}$] {};
\node (a3) at (2,0) [S,label=right:${P_3(x_2,x_3,x_5)}$] {};
\node (a4) at (4,0) [S,label=right:${P_4(x_1,x_4,x_5)}$] {};
\node (a5) at (6,0) [S,label=right:${P_5(x_2,x_3,x_5)}$] {};
\node (a6) at (8,0) [S,label=right:${P_6(x_1,x_2,x_3)}$] {};
\node (a7) at (10,0) [S,label=right:${P_7(x_1,x_4,x_5)}$] {};
\draw (a1.west) to (1.east) [thick];
\draw (a1.west) to (3.east) [thick];
\draw (a1.west) to (5.east) [thick];
\draw (a2.west) to (1.east) [thick];
\draw (a2.west) to (2.east) [thick];
\draw (a2.west) to (4.east) [thick];
\draw (a3.west) to (2.east) [thick];
\draw (a3.west) to (3.east) [thick];
\draw (a3.west) to (5.east) [thick];
\draw (a4.west) to (1.east) [thick];
\draw (a4.west) to (4.east) [thick];
\draw (a4.west) to (5.east) [thick];
\draw (a5.west) to (2.east) [thick];
\draw (a5.west) to (3.east) [thick];
\draw (a5.west) to (5.east) [thick];
\draw (a6.west) to (1.east) [thick];
\draw (a6.west) to (2.east) [thick];
\draw (a6.west) to (3.east) [thick];
\draw (a7.west) to (1.east) [thick];
\draw (a7.west) to (4.east) [thick];
\draw (a7.west) to (5.east) [thick];
\end{tikzpicture}
\end{center}
\caption{A bipartite graph representation of Goldreich's generator.}
\label{fig:goldreich}
\end{figure}

\subsection{(r-1)-wise independence and MDS codes}
In the following, several properties of local functions are defined which are useful for analysis of some attackers to
Goldreich's generator.

\begin{definition}[$(r-1)$-wise independence]
Let $X_1,\dotsc,X_k$ be uniformly distributed random variables on $[q]$.
The function $f$ is said to be $(r-1)$-wise independent 
if for any $T\subseteq [k]$ of size at most $r-1$,
the distribution of $(X_i)_{i\in T}$ and $f(X_1,\dotsc,X_k)$ is uniform on $[q]^{|T|}\times [q]^l$.
\end{definition}

\begin{example}
We can regard the alphabet $[q]$ as an Abelian group $\mathbb{Z}/q\mathbb{Z}$.
A function $f(\bm{x})=x_1+x_2+\dotsb+x_r + f'(x_{r+1},\dotsc,x_{k})$
is $(r-1)$-wise independent for any $f'\colon [q]^{k-r}\to [q]$.
\end{example}

Let us consider uniformly distributed random variables $X_1,\dotsc,X_k$ on $[q]$.
If the function $f$ is $(r-1)$-wise independent, even if one knows values of the output $f(X_1,\dotsc,X_k)$ and 
$r-2$ of the $k$ input random variables,
one cannot guess any one of the other unknown input variable for $f$.
This property is useful for deriving lower bounds for some algorithms trying to invert Goldreich's generators.
On the other hand, if one knows the output of $(r-1)$-wise independent function and $r-1$ variables in the input,
one may guess one of the other unknown input variable.
We can consider the extremal $(r-1)$-wise independent functions
 for which the output and $r-1$ of the input variables uniquely fix all of the other $k-r+1$ input variables.

\begin{definition}[Maximum distance separable code]
A subset $M\subseteq [q]^{k}$ is said to be maximum distance separable (MDS) code of the dimension $r-1$
if $|M|=q^{r-1}$ and $d=k-r+2$ where $d$ is the minimum distance of the code $M$, i.e.,
$d:=\min_{\bm{x}\in M, \bm{x}'\in M, \bm{x}\ne\bm{x}'}|\{i\in[k]\mid x_i\ne x'_i\}|$.
\end{definition}

\begin{example}[Trivial MDS codes]\label{exm:trivialMDS}
$M_{\mathrm{SPC}}=\{(x_1,\dotsc,x_k)\in[q]^k\mid x_1+\dotsb+x_k=0\}$ is an MDS code
of the dimension $k-1$.
The support of uniquely extendible constraint is also an MDS code with the above parameters~\cite{connamacher2004exact}.
$M_{\mathrm{R}}=\{(x_1,\dotsc,x_k)\in[q]^k\mid x_1=x_2=\dotsb=x_k\}$ is an MDS code
of the dimension $1$.
MDS codes of the dimension $k-1$ and of the dimension $1$ are called trivial MDS codes.
\end{example}

Assume that for given $f\colon[q]^k\to[q]^{k-r+1}$,
an inverse image $f^{-1}(\bm{y})\subseteq [q]^k$ is an MDS code of the dimension $r-1$ for any $\bm{y}\in[q]^{k-r+1}$.
In that case, we say that a local function $f$ has MDS inverse image.
If $f$ has MDS inverse image, then
for any $\bm{y}\in[q]^{k-r+1}$, for any $S\subseteq[k]$ of size $r-1$ and
for any $(x_i)_{i\in S}\in[q]^{|S|}$,
 there exists one and only one $(x_i)_{i\in [k]\setminus S}\in[q]^{k-|S|}$ 
such that $\bm{x}\in f^{-1}(\bm{y})$ since the minimum distance of the inverse image $f^{-1}(\bm{y})$ is $k-r+2$
 and since the dimension of the inverse image $f^{-1}(\bm{y})$ is $r-1$.
Hence, $f$ is $(r-1)$-wise independent and also extremal one, i.e., the output of $f$ and $r-1$ of input variables
uniquely determine the other $k-r+1$ input variables.

For a prime power $q$, one can regard the alphabet $[q]$ as a finite field $\mathbb{F}_q$.
If one has a linear MDS code $M$ of length $k$ and dimension $r-1$, 
it is easy to construct a local function $f\colon [q]^k\to [q]^{k-r+1}$ with MDS inverse image as follows.
Let $H$ be a $(k-r+1)\times k$ parity-check matrix for $M$, i.e., $M=\{\bm{x}\in\mathbb{F}_q^k\mid H\bm{x}=0\}$.
Then, the function $f(\bm{x}):= H\bm{x}$ has MDS inverse image since $f^{-1}(\bm{y})=\{\bm{x}\in\mathbb{F}_q^k\mid H\bm{x}=\bm{y}\}$
is an MDS code of the dimension $r-1$ for any $\bm{y}\in \mathbb{F}^{k-r+1}$.
Note that one can also construct a local function with MDS inverse image for any (not necessarily prime power) $q$
 by using an (not necessarily linear) MDS code of length $k+l$.

For $r=2$ and $r=k$, there exist trivial MDS codes for arbitrary $q$ as in Example~\ref{exm:trivialMDS}.
However, for $r\in\{3,4,\dotsc,k-1\}$,
the existence of MDS codes for given parameters $k$ and $q$ is not obvious.
It is known that, if $r-1\ge q+1$, then MDS code does not exist unless $k=r$.
When $r-1\le q$, it is conjectured that except for some special cases, all MDS codes have codelength $k$ at most $q+1$.
The doubly extended Reed-Solomon codes defined on a finite field are MDS codes with codelength $q+1$ which can have arbitrary dimension.
Hence, for given $k$ and $r\in\{3,4,\dotsc, k-1\}$ and for a prime power $q\ge k-1$,
there exists an MDS code of arbitrary dimension on the alphabet $[q]$.

\section{Randomly generated Goldreich's generator and main result}
For analyzing Goldreich's generator, randomly generated bipartite graph has been considered since randomly generated bipartite graph
is expander with high probability.
The random ensemble of Goldreich's generator is defined as follows.
The alphabet size $q$, the input length $n$ of the whole generator, the number $m$ of local functions,
 the input length $k$ of local functions and the output length $l$ of local functions are fixed.
The local functions $f_a\colon [q]^k\to [q]^{l}$ are also fixed for every $a\in\{1,2,\dotsc,m\}$.
The above parameters and local functions are given and not randomly generated.
Then, 
$k$-tuple $(i^{(a)}_{1},\dotsc, i^{(a)}_{k})$ of distinct indices of variables are uniformly chosen from $\{1,\dotsc,n\}$
for each $a\in\{1,2,\dotsc,m\}$ independently. 
Then, the generator is defined as $(f_1(x_{i^{(1)}_1},\dotsc, x_{i^{(1)}_k}),\dotsc, f_m(x_{i^{(m)}_1},\dotsc, x_{i^{(m)}_k}))$
as in Section~\ref{subsec:goldreich}.
Hence, in this random ensemble, a Goldreich's generator is uniformly chosen from all of the $[n(n-1)\dotsm(n-k+1)]^m$ possible choices.

We also consider randomly generated planted $k$-CSP as follows.
A Goldreich's generator $g$ is randomly generated in the above way.
The input values $\bm{x}^*\in[q]^n$ is uniformly chosen from $[q]^n$.
Then, the output $\bm{y}\in[q]^{lm}$ of $g$ for the input $\bm{x}^*$ is computed.
A pair of Goldreich's generator $g$ and the output $\bm{y}\in[q]^{lm}$ is an instance of randomly generated planted $k$-CSP.
We regard $g(\bm{x})=\bm{y}$ as a $k$-CSP for variables $\bm{x}\in[q]^n$.
Here, the input values $\bm{x}^*$ is called a planted assignment, planted configuration or planted solution.

For attacking PRG generated by Goldreich's generator, we can consider the following strategy.
Let $L_g:=\{\bm{y}\in[q]^{lm}\mid \exists \bm{x}\in[q]^n, g(\bm{x})=\bm{y}\}$.
Then, an attacker $h_g$ is defined as
\begin{align*}
h_g(\bm{y})= 1 &\iff \bm{y}\in L_g\\
h_g(\bm{y})= 0 &\iff \bm{y}\notin L_g.
\end{align*}
If the input for $h_g$ is generated by $g$, it always returns 1, i.e., $\Pr(h_g(g(X))=1)=1$ where $X$ is a uniform random variable on $[q]^n$.
If the input for $h_g$ is uniform random variable on $[q]^{lm}$, 
the probability that $h_g$ returns 1 is $\Pr(h_g(U)=1)=|L_g|/q^m\le q^{n-m}$, which is small if $m$ is much larger than $n$ where $U$ is a uniform random variable on $[q]^m$.
From the above observation, Goldreich's generator is not secure against the attack $h_g$, which is of course not necessarily efficiently computable.
In this paper, we try to find a certificate $\bm{x}\in g^{-1}(\bm{y})$ for $\bm{y}\in L_g$ for implementing $h_g$.
There is also another attack called a ``correlation attack'', which tries to find a certificate for $\bm{y}\notin L_g$.
Generally it is difficult to show that $L_g$ does not have a polynomial-time algorithm since $L_g$
has a short certificate $\bm{x}\in[q]^n$ and has a verifier in $\mathsf{NC}^0$.

As mentioned in the previous section, the concept of $(r-1)$-wise independence expresses security of
randomly generated Goldreich's generator as OWFs and PRGs for some algorithms.
The following results for $r\ge 3$ are known for randomly generated Goldreich's generator on the binary alphabet for $l=1$.

\begin{lemma}[\cite{mossel2006varepsilon}, \cite{v003a002}, \cite{odonell2014goldreich}]
Assume that all local functions are not $r$-wise independent for some $r\ge 3$.
If $m=Cn^{r/2}\sqrt{\log n}$ for sufficiently large constant $C$, then there is a polynomial-time algorithm inverting
the randomly generated Goldreich's generator with high probability.
If $m=Cn^{r/2}$ for sufficiently large constant $C$, then there is a polynomial-time algorithm distinguishing
the output of the randomly generated Goldreich's generator from uniform random variables with high probability.
\end{lemma}

\begin{lemma}[\cite{odonell2014goldreich}, \cite{v008a012}]
Assume that all local functions are $(r-1)$-wise independent for some $r\ge 3$.
If $m=O(n^{r/2-\delta})$ for some $\delta>0$,
 then a semidefinite programming relaxation with high probability cannot fix any input variable of 
the randomly generated Goldreich's generator with small modification.
\end{lemma}

From the above results, randomly generated Goldreich's generator is insecure when $m=O(n^{\frac{r}{2}+\delta})$
and secure against SDP when $m=O(n^{\frac{r}{2}-\delta})$ for any $\delta>0$.
Hence, a local function with large $r$ seems to be preferable.
However, for the binary alphabet, $r=k$ holds only when the predicate is affine.
In that case, one can efficiently find the input assignment by solving the system of linear equations.
Generally, the degree $d$ of polynomial representation on $\mathbb{F}_2$ of $(r-1)$-wise independent predicate is at most $k-r$ if $r\le k-2$.
If predicates in the generator have the degree $d$ in the polynomial representation, there is a linear attack when $m=O(n^d)$~\cite{mossel2006varepsilon}.
Hence, for the binary alphabet, large $r$ implies small $d$ which means there exists another attacker.
However, for non-binary cases, $r=k$ does not immediately imply existence of another attacker.
While $r=k$ implies that the constraint given by a local function is uniquely extendible constraint, it is known that combination of three types of
uniquely extendible constraints on quaternary alphabet can represent the three coloring problem on a graph~\cite{connamacher2004exact}.
This observation gives a motivation for considering non-binary generalization.

The followings are the main results of this paper 
on simple LP relaxation attack for Goldreich's generator
 using $u(n)\in \omega(1)\cap o(n)$ known input variables which will be defined in the next section.

\begin{theorem}\label{thm:main0}
Assume that all local functions are $(r-1)$-wise independent for some $r\ge 3$.
If $m=\mu \frac{n^{r-1}}{u(n)^{r-2}}$ for arbitrary constant $\mu < \mu_{\mathrm{c}}(k,r)$,
 then a linear programming relaxation using $u(n)$ known input variables cannot
 fix $\frac1{r-2} u(n)$ input variables of the randomly generated Goldreich's generator with probability exponentially close to 1 with respect to $u(n)$.
\end{theorem}

\begin{theorem}\label{thm:main1}
Assume that all local functions have MDS inverse image of the dimension $r-1$ for some $r\ge 3$.
If $m=\mu \frac{n^{r-1}}{u(n)^{r-2}}$ for any constant $\mu > \mu_{\mathrm{c}}(k,r)$,
 then a linear programming relaxation using $u(n)$ known input variables can fix
 linearly many number of input variables of the randomly generated Goldreich's generator with probability exponentially close to 1 with respect to $u(n)$.
\end{theorem}

In the above results, $u(n)\in\omega(1)\cap o(n)$ variables are assumed to be known.
This assumption is justified when $u(n)=O(\log n)$ since one can 
try to apply the LP for all of the $q^{u(n)}$ assignments in polynomial time in $n$.
If linearly many variables are fixed without contradiction and if $m$ is superlinear,
 one can distinguish the output of Goldreich's generator from uniform random variable
since with high probability, linearly many variables cannot be fixed without contradiction when uniform random variables are assigned to
the output variables of local functions.
When $u(n)=O(\log n)$, from Theorem~\ref{thm:main0}, the LP relaxation attack fails with probability polynomially close to 1 with respect to $n$.
Hence, it is not sufficient to claim that the generator is secure in the strong sense against the LP relaxation attack~\cite{books/cu/Goldreich2001}.
When $u(n)=O(\log n)$, the order $\frac{n^{r-1}}{u(n)^{r-2}}$ of $m$ in the above results is much larger than $n^{\frac{r}{2}}$,
which means that the LP relaxation attack is suboptimal.
By tightening the polytope, one can obtain the currently optimal order $n^{\frac{r}{2}}$.
This problem will be discussed in Section~\ref{sec:condis}.
When $u(n)=\omega(\sqrt{n})$, Goldreich's generator is no longer secure even when $m=o(n^{r/2})$ since $\frac{n^{r-1}}{u(n)^{r-2}}=o(n^{r/2})$.
In this case, in our knowledge, there is no attack inverting Goldreich's generator asymptotically better than the simple LP relaxation attack in this paper.

\section{Linear programmings, peeling algorithms and stopping sets}\label{sec:LPMP}
In this section, the LP relaxation for $k$-CSP which is the algorithm discussed in this paper is defined.
Furthermore, it is shown that the LP relaxation is related to combinatorial structures called stopping sets.
A $k$-CSP can be represented as the integer programming (IP)
\begin{equation}\label{eq:oIP}
\min_{\bm{x}\in [q]^n}: \sum_{a=1}^m \mathbb{I}\left\{f_a(\bm{x}^{(a)}) \ne \bm{y}_a\right\}
\end{equation}
where $\bm{x}^{(a)}:= (x_i)_{i\in \partial(a)}$ for $\partial(a):=\{i^{(a)}_1,\dotsc,i^{(a)}_k\}$
and where $\mathbb{I}\{\cdot\}$ is the indicator function.
For applying the LP relaxation to the IP, we first choose $u(n)\in \omega(1)\cap o(n)$ variables and generate $q^{u(n)}$ sub-IPs for each assignment on the $u(n)$ variables
since the LP relaxation for~\eqref{eq:oIP} has a trivial useless solution which will be mentioned in this section.
Without loss of generality, we can assume that $x_1,\dotsc,x_{u(n)}$ are fixed to some assignment since we consider uniform random construction of the $k$-CSP.
Then, for each assignment $(z_1,\dotsc,z_{u(n)})\in [q]^{u(n)}$, we consider the sub-IP
\begin{equation}\label{eq:IP}
\min_{\substack{\bm{x}\in [q]^n\\ x_i=z_i \text{ for } i=1,2,\dotsc,u(n)}}: \sum_{a=1}^m \mathbb{I}\left\{f_a(\bm{x}^{(a)}) \ne \bm{y}_a\right\}.
\end{equation}
The original IP~\eqref{eq:oIP} can be solved by taking minimum among all of the $q^{u(n)}$ sub-IPs.
Note that when $u(n)=O(\log n)$, there are polynomially many sub-IPs.
The LP relaxation is applied for the sub-IPs~\eqref{eq:IP} rather than the original IP~\eqref{eq:oIP}.
First, the marginal polytope, which gives a tight LP relaxations, is defined as follows.
\begin{definition}[Marginal polytope]
Let $\mathrm{DIST}_n:=\{(p(\bm{x}))_{\bm{x}\in [q]^n}\mid p(\bm{x})\ge 0,\,\forall\bm{x}\in [q]^n, \sum_{\bm{x}}p(\bm{x})=1\}$.
Then, the marginal polytope is defined as
\begin{align*}
\mathrm{MARG}&:=\Bigl\{((p_i(x_i))_{i\in [n], x_i\in [q]},(p_{(a)}(\bm{x}^{(a)}))_{a\in [m], \bm{x}^{(a)}\in [q]^{k}})\mid
 \exists p\in \mathrm{DIST}_n, p_i(x_i)=\sum_{\bm{x}\setminus x_i}p(\bm{x}),\,\forall i\in [n], \forall x_i\in [q],\\
&\qquad p_{(a)}(\bm{x}^{(a)})=\sum_{\bm{x}\setminus\bm{x}^{(a)}} p(\bm{x}),\,\forall a\in [m], \forall \bm{x}^{(a)}\in [q]^{k}\Bigr\}.
\end{align*}
\end{definition}
Using the marginal polytope, one obtains the tight LP relaxation for~\eqref{eq:IP}
\begin{equation*}
\min_{\substack{((p_i),(p_{(a)}))\in\mathrm{MARG},\\ p_i(z_i)=1 \text{ for } i=1,2,\dotsc,u(n)}}: 
\sum_{a=1}^m \mathbb{E}_{p_{(a)}}\left[f_a(\bm{X}^{(a)}) \ne \bm{y}_a\right].
\end{equation*}
Although the above LP relaxation using the marginal polytope is tight, the marginal polytope uses exponentially many variables and inequalities.
Hence, we consider loose but more efficient LP relaxation.

\begin{definition}[Local marginal polytope]
The local marginal polytope is defined as
\begin{align*}
\mathrm{LOCAL}:=\Bigl\{((b_i(x_i))_{i\in [n], x_i\in [q]}, (b_{(a)}(\bm{x}^{(a)}))_{a\in [m], \bm{x}^{(a)}\in [q]^{k}})\mid\,
& b_{(a)}\in \mathrm{DIST}_{k},\,\forall a\in [m],\\
& b_i(x_i)=\sum_{\bm{x}^{(a)}\setminus x_i} b_{(a)}(\bm{x}^{(a)}),\,\forall i\in[n],a\in[m] \text{ satisfying } i\in\partial(a)\Bigr\}.
\end{align*}
\end{definition}
Obviously, it holds $\mathrm{MARG}\subseteq\mathrm{LOCAL}$ from the definition.
If the bipartite graph representing the $k$-CSP problem is tree, it holds $\mathrm{MARG}=\mathrm{LOCAL}$~\cite{wainwright2008graphical}.
However, generally, the inclusion relation is strict.
The LP relaxation using the local marginal polytope for~\eqref{eq:IP} is obtained as
\begin{equation}\label{eq:LP}
\min_{\substack{((b_i),(b_{(a)}))\in\mathrm{LOCAL},\\ b_i(z_i)=1,\forall i=1,2,\dotsc,u(n)}}: 
\sum_{a=1}^m \mathbb{E}_{p_{(a)}}\left[f_a(\bm{X}^{(a)}) \ne \bm{y}_a\right].
\end{equation}
Note that if we consider the LP relaxation using the local marginal polytope for~\eqref{eq:oIP},
\begin{equation}\label{eq:dLP}
\min_{((b_i),(b_{(a)}))\in\mathrm{LOCAL}}: 
\sum_{a=1}^m \mathbb{E}_{p_{(a)}}\left[f_a(\bm{X}^{(a)}) \ne \bm{y}_a\right]
\end{equation}
there always exists a trivial solution having zero objective value when all local functions are 1-wise independent,
which is $b_i(x_i)=1/q$ for all $i\in [n]$ and $x_i\in [q]$, and $b_{(a)}(\bm{x}^{(a)})=\mathbb{I}\{f_a(\bm{x}^{(a)})= \bm{y}_a\}/|f_a^{-1}(\bm{y}_a)|$ for all $a\in [m]$
and all $\bm{x}^{(a)}\in [q]^{k}$.
Hence, we consider the LP relaxation for~\eqref{eq:IP} rather than that for~\eqref{eq:oIP}.
If all marginals in $\mathrm{LOCAL}$ are deterministic, we call it an integral assignment.
A subset $\{i\in[n]\mid  \exists x_i\in [q], b_i(x_i)=1\}\subseteq [n]$ is called
 an integral part of $((b_i)_{i\in[n]}, (b_{(a)})_{a\in[m]})\in\mathrm{LOCAL}$.
\begin{remark}
The LP using $\mathrm{LOCAL}$ can be regarded as the minimization of the Bethe free energy at zero-temperature in statistical physics and 
is sometimes called the basic LP in computer science~\cite{thapper2012power}.
On the other hand, there is a message passing algorithm called belief propagation (BP) which tries to minimize the Bethe free energy on non-zero temperature~\cite{yedidia2005constructing}.
Recently, spectral algorithm using non-backtracking matrix is proposed in~\cite{Krzakala25112013}, which can be regarded as linearization of BP
on trivial fixed point.
These three algorithms can be understood by the Bethe approximation.
\end{remark}

In this paper, we consider the limit of solvability of the randomly generated planted $k$-CSP by the LP relaxation~\eqref{eq:LP}.
If there exists an optimal solution for~\eqref{eq:LP} with zero objective value, we call it a zero-optimal solution.
If the $u(n)$ fixed assignment cannot be extended to a solution of the planted $k$-CSP problem~\eqref{eq:oIP},
there does not exist an integral zero-optimal solution for~\eqref{eq:LP}.
If the $u(n)$ fixed assignment can be extended to a solution of~\eqref{eq:oIP},
there exist an integral zero-optimal solution although there also exist non-integral zero-optimal solutions in general.
In the rest of this paper, we assume that the $u(n)$ variables are fixed to be the values of the planted assignment $\bm{x}^*\in[q]^n$,
and consider whether the LP~\eqref{eq:LP} has non-integral zero-optimal solutions.

It is well-known that for $k$-CSP problems including zero-optimal solutions, the LP~\eqref{eq:LP} using the local marginal polytope
is strongly related to a simple message-passing algorithm.
\begin{definition}[Peeling algorithm for a bipartite graph]
For $d\ge 2$, the $d$-peeling algorithm for a bipartite graph starting from a set $V\subseteq[n]$ of variables is defined as follows.
First, all variable vertices not in $V$ are removed from the bipartite graph.
Then, if there is a constraint vertex of the degree at most $d-1$,
then the constraint vertex and all of the at most $d-1$ variable vertices connected to the constraint vertex
are removed from the bipartite graph.
This process is iterated until there is no constraint vertex of the degree at most $d-1$. 
\end{definition}
The $d$-peeling algorithm naturally appears in many problems,
 e.g., the decoding of low-density parity-check codes~\cite{Luby:1997:PLC:258533.258573},
 the satisfiability and clustering phase transition of random $k$-XORSAT~\cite{ibrahimi2011set} and cuckoo hashing~\cite{dietzfelbinger2010tight}.
The peeling algorithm starting from $V$ stops if and only if the current set variables forms a structure called a stopping set.
This type of peeling algorithm on the same random graph ensemble was considered for $k=3$ and $d=2$ in~\cite{coja2012propagation}, \cite{watanabe2013mp}.
\begin{definition}[Stopping set~{\cite{di2002finite}}]
For $d\ge 2$, a subset $V'\subseteq [n]$ is called a $d$-stopping set if $|\partial (a)\cap V'|\notin \{1,2,\dotsc,d-1\}$ for all constraint $a\in [m]$.
\end{definition}
It is obvious that the $d$-peeling algorithm starting from $V$
 removes all variables if and only if there does not exist non-empty $d$-stopping set included by $V$.
The concept of $d$-stopping set is useful for analyzing the LP~\eqref{eq:LP}.

\begin{lemma}[{\cite{feldman2005using}}]\label{lem:sslp}
Assume that all local functions are $(r-1)$-wise independent.
If there is a $(k-r+2)$-stopping set $V'\subseteq \{u(n)+1,u(n)+2,\dotsc,n\}$,
then the LP~\eqref{eq:LP} using the local marginal polytope has a zero-optimal solution whose integral part is $[n]\setminus V'$.
\end{lemma}
\begin{proof}
let $V'\subseteq \{u(n)+1,u(n)+2,\dotsc,n\}$ be a $(k-r+2)$-stopping set.
Then, $((b_i),(b_{(a)}))$ defined by
$b_i(x_i)=1/q$ for any $i\in V'$, $x_i\in[q]$,
$b_i(x^*_i)=1$ for any $i\notin V'$ and
\begin{align*}
b_{(a)}(\bm{x}^{(a)})&=\frac1{|f_a^{-1}(\bm{y}_a)|},\hspace{2em} \forall \bm{x}^{(a)}\in [q]^{k} \quad\text{ satisfying }\quad
x_i= x^*_i \quad \forall i\notin \partial (a)\cap V', f_a(\bm{x}^{(a)})=\bm{y}_a
\end{align*}
is an element of $\mathrm{LOCAL}$ from the assumption of $(r-1)$-wise independence of $f_a$.
The above $((b_i),(b_{(a)}))$ has a zero objective value.
\end{proof}

For local functions having MDS inverse image, the converse of Lemma~\ref{lem:sslp} also holds.

\begin{lemma}[{\cite{feldman2005using}}]
\label{lem:lpss}
Assume that all local functions have MDS inverse of the dimension $r-1$.
If the LP~\eqref{eq:LP} using the local marginal polytope has a zero-optimal solution whose integral part is $V'\subseteq [n]$, then
$[n]\setminus V'$ is $(k-r+2)$-stopping set.
\end{lemma}
\begin{proof}
Let $((b_i),(b_{(a)}))\in\mathrm{LOCAL}$ be one of the zero-optimal solution (not necessarily integral) for the LP.
For any $a\in [m]$, $b_{(a)}(\bm{x}^{(a)})>0$ only for $\bm{x}^{(a)}$ satisfying the constraint $f_a$
since $((b_i),(b_{(a)}))$ has a zero objective value.
Hence, if $|\partial (a)\cap V'|\ge r-1$, then $|\partial (a)\cap V'|$ must be $k$.
\end{proof}

From Lemmas~\ref{lem:sslp} and \ref{lem:lpss}, it is sufficient to analyze $(k-r+2)$-stopping sets
for showing Theorems~\ref{thm:main0} and \ref{thm:main1}.
A non-empty $d$-stopping set of size smaller than $\alpha n$ is called $\alpha$-small $d$-stopping set for $\alpha\in(0,1)$.
A non-empty $d$-stopping set which is not $\alpha$-small $d$-stopping set is called $\alpha$-large $d$-stopping set.

\begin{lemma}[Threshold for small stopping set]\label{lem:small}
For any $\alpha\in(0,1)$,
the randomly generated bipartite graph including $m=\mu n\log n$ constraint does not have $\alpha$-small $d$-stopping set
with probability $1-O(n^{1-\mu k+\epsilon})$
for any
 $\mu>1/k$
and $\epsilon\in(0,\mu k-1)$.
\end{lemma}
\begin{proof}
Let $Z(l)$ be a random variable corresponding to the number of $d$-stopping sets of size $l$ for the randomly generated bipartite graph.
Then, the probability that randomly generated bipartite graph includes $d$-stopping set of size at most $\alpha n$ is
upper bounded by Markov's inequality
\begin{equation*}
\Pr\left(\sum_{l=1}^{\alpha n}Z(l)\ge 1\right)
\le
\sum_{l=1}^{\alpha n} \mathbb{E}[Z(l)].
\end{equation*}
The expected number of $d$-stopping sets of size $l$ is simply represented as
\begin{equation*}
\mathbb{E}[Z(l)]
=
\binom{n}{l}
\left(\sum_{s=0, d,d+1,\dotsc,k} \frac{\binom{l}{s}\binom{n-l}{k-s}}{\binom{n}{k}}\right)^m.
\end{equation*}
Especially for $d=2$, it holds
\begin{equation*}
\mathbb{E}[Z(l)]
=
\binom{n}{l}
\left(1-\frac{l\binom{n-l}{k-1}}{\binom{n}{k}}\right)^m.
\end{equation*}
When $m=\gamma n$ for some constant $\gamma>0$, it holds
\begin{equation*}
\frac1n\log \mathbb{E}[Z(\delta n)]
=
h(\delta)
+
\gamma \log \left(1-k \delta(1-\delta)^{k-1}\right)
+ o(1)
\end{equation*}
for any $\delta\in(0,1)$ where $h$ denotes the binary entropy function.
Hence, for any fixed $\delta\in(0,\alpha)$, there is a constant $\gamma_\delta$ such that
\begin{equation*}
h(\delta')+
\gamma_\delta \log \left(1-k \delta'(1-\delta')^{k-1}\right)
\le -1
\end{equation*}
for any $\delta'\in[\delta,\alpha]$.
Hence, 
\begin{equation*}
\sum_{l=\delta n}^{\alpha n}\mathbb{E}[Z(l)]
\le n\exp\{-n\}
\end{equation*}
when $m=\gamma_\delta n$.
From an inequality
\begin{equation*}
\log\left(1-\frac{l\binom{n-l}{k-1}}{\binom{n}{k}}\right)
\le
-\frac{l\binom{n-l}{k-1}}{\binom{n}{k}}
\end{equation*}
one obtains for $m=\mu n\log n$ that
\begin{align*}
\sum_{l=1}^{\delta n}\mathbb{E}[Z(l)]
&\le
\sum_{l=1}^{\delta n}\binom{n}{l}\exp\left\{-m\frac{l\binom{n-\delta n}{k-1}}{\binom{n}{k}}\right\}
\le\left(1+\exp\left\{-m\frac{\binom{n-\delta n}{k-1}}{\binom{n}{k}}\right\}\right)^n-1\\
&=
\left(1+n^{-\mu k(1-\delta)^{k-1} + o(1)}\right)^n-1
\end{align*}
for any $\delta\in(0,1)$.
Let $\delta_\mu := 1-1/(\mu k)^{1/(k-1)}$.
For any $\mu>1/k$ and any $\delta \in (0,\delta_\mu)$, it holds $\mu k(1-\delta)^{k-1}>1$,
i.e.,
\begin{equation*}
\left(1+n^{-\mu k(1-\delta)^{k-1} + o(1)}\right)^n-1
=O\left(n^{1-\mu k(1-\delta)^{k-1}}\right).
\qedhere
\end{equation*}
\end{proof}
Conversely if $m=\mu n\log n$ for $\mu<1/k$, from the theory of the coupon collector's problem,
with high probability there exists a variable which is not included in any constraint.
Hence, there exists a $d$-stopping set of size 1 with high probability.
For $\alpha$-large $k$-stopping set, the threshold is obtained as follows.
\begin{lemma}\label{lem:2large}
For any $\mu>(k(k-1))^{-1}$, there exists $\alpha\in(0,1)$
such that
the randomly generated bipartite graph including $m=\mu n$ constraints does not have
$\alpha$-large $k$-stopping set in $\{u(n)+1,\dotsc,n\}$ with probability exponentially close to 1 with respect to $u(n)$.
\end{lemma}
\begin{proof}
From the theory of random hypergraphs, if $m=\mu n$ where $\mu> (k(k-1))^{-1}$, then the random hypergraph has a giant component,
which is a connected component of size proportional to $n$,
 with probability tends to 1 exponentially fast as $n\to\infty$~\cite{schmidt1985component}, \cite{behrisch2007local}.
It is also shown in~\cite{RSA:RSA20160} that the size of giant component is concentrated around $(1-\rho)n$ where $\rho\in(0,1)$ is the unique solution of
\begin{equation*}
\rho=\exp\{\mu k (\rho^{k-1}-1)\}.
\end{equation*}
Hence, the probability that the size of giant component is greater than $(1-\rho-\delta)n$ tends to 1 exponentially fast with respect to $n$
for any $\delta>0$.
In that case, the probability that all of the $u(n)$ known variables are not included in the giant component is at most
$(\rho+\delta)^{u(n)}$.
\end{proof}

From Lemmas~\ref{lem:lpss},~\ref{lem:small} and \ref{lem:2large}, 
if all local functions have MDS inverse image of dimension $r-1$ for $r=2$,
then the randomly generated Goldreich's generator including $m=\mu n\log n$ local functions is inverted 
by the LP~\eqref{eq:LP} with high probability for any $\mu > 1/k$.
The converse of Lemma~\ref{lem:2large} is also obtained as follows.
\begin{lemma}\label{lem:2ML}
For any $\mu<(k(k-1))^{-1}$,
the randomly generated bipartite graph including $m=\mu n$ constraints has
$k$-stopping set of size larger than $n-(1+\tau) u(n)$ included in $\{u(n)+1,\dotsc,n\}$ with high probability
for any $\tau$ strictly larger than
\begin{equation*}
\frac{k(k-1)\mu}{1-k(k-1)\mu}.
\end{equation*}
\end{lemma}
For $r=2$, the condition $m\ge(1/k+\delta)n\log n$ for vanishing small stopping sets is stronger than
the condition $m\ge ([k(k-1)]^{-1}+\delta) n$ for vanishing large stopping sets.
On the other hand, for $r\ge 3$, the situation is different unless $u(n)$ is quite large,
i.e., $u(n)=\Omega(n/(\log n)^{1/(r-2)})$.

\begin{theorem}\label{thm:ML}
Fix $r\ge 3$.
For any constant $\mu<\mu_\mathrm{c}(k,r)=\binom{k}{r}^{-1}\frac{(r-2)^{r-2}}{r(r-1)^{r-1}}$,
the randomly generated bipartite graph including $m=\mu \frac{n^{r-1}}{u(n)^{r-2}}$ constraints has
$(k-r+2)$-stopping set of size larger than $n-(1+\tau) u(n)$ included in $\{u(n)+1,\dotsc,n\}$ with high probability $p(n, \mu, \tau)$
for any $\tau>\tau^*$ where $\tau^*\in(0,1/(r-2))$ is the unique solution in $(0,1/(r-2))$ of
\begin{equation*}
\mu = \frac1{\binom{k}{r}}\frac{\tau^*}{r(1+\tau^*)^{r-1}}.
\end{equation*}
Here, the probability $p(n,\mu,\tau)$ is at least
 $1-\exp\{\inf_{\lambda>0,\tau'\in(0,\tau)}\allowbreak \varphi_{k,r}(\mu,\lambda,\tau') u(n) + O(\max\{1,u(n)^2/n\})\}$
for
\begin{equation}\label{eq:Ek}
\varphi_{k,r}(\mu,\lambda,\tau):= \mu \left(\exp\{(k-r+1)\lambda\}-1\right)\binom{k}{r-1}(1+\tau)^{r-1}- \lambda\tau.
\end{equation}
\end{theorem}

From Theorem~\ref{thm:ML} and Lemma~\ref{lem:sslp}, for any constant $\mu < \mu_\mathrm{c}(k,r)$,
the LP relaxation~\eqref{eq:LP} with the $u(n)$ correctly assigned variables for the planted $k$-CSP
problem including $m=\mu\frac{n^{r-1}}{u(n)^{r-2}}$ constraints has a solution whose integral part is at most $(1+\tau)u(n)$ with high probability.
Hence, Theorem~\ref{thm:main0} is shown.
The converse is also obtained as follows.

\begin{theorem}\label{thm:MU}
Fix $r\ge 3$.
For any $\alpha\in(0,1)$ and
for any constant $\mu>\mu_\mathrm{c}(k,r)$,
the randomly generated bipartite graph including $m=\mu \frac{n^{r-1}}{u(n)^{r-2}}$ constraints does not have
$\alpha$-large $(k-r+2)$-stopping set included in $\{u(n)+1,\dotsc,n\}$ with probability at least
 $1-\exp\{\sup_{\tau >0}\inf_{\lambda<0}\allowbreak \varphi_{k,r}(\mu,\lambda,\tau) u(n) + O(\max\{u(n)^2/n,\,\log u(n)\})\}$.
\end{theorem}
From Theorem~\ref{thm:MU} and Lemma~\ref{lem:lpss}, one obtains Theorem~\ref{thm:main1}.
The above results on thresholds of stopping sets are summarized in Table~\ref{tbl:ss}.
Lemma~\ref{lem:2ML} and Theorems~\ref{thm:ML} and \ref{thm:MU} are proved in the following sections.

\begin{table}
\renewcommand{\arraystretch}{1.3}
\caption{Thresholds of the number of constraints for $(k-r+2)$-stopping sets}
\label{tbl:ss}
\centering
\begin{tabular}{c|c|c}
\hline
& Small & Large\\
\hline
$r=2$ & $\frac1{k} n\log n$& $\frac1{k(k-1)}n$\\
\hline
$r\ge 3$ & $\frac1{k} n\log n$& $\frac1{\binom{k}{r}} \frac{(r-2)^{r-2}}{r(r-1)^{r-1}} \frac{n^{r-1}}{u(n)^{r-2}}$
\end{tabular}
\end{table}

\section{Evolution on number of equations for the peeling algorithm}
\subsection{The Markov chain}
In this section, we consider the numbers of constraints with particular degree at each step of the iterations of the $(k-r+2)$-peeling algorithm
starting from $\{u(n)+1,\dotsc,n\}$ on the randomly generated bipartite graph, which is essentially equivalent to
Erd\H{o}s-R\'{e}nyi hypergraph~\cite{RSA:RSA20160}.
For the analysis, we assume that only one constraint $a\in[m]$ of the degree at most $k-r+1$ is chosen in each step
and that one of the variable connected to the constraint $a$ is removed
from the bipartite graph.
Let $C_j(t)$ be a random variable corresponding to the number of constraint vertices of degree $j$ after $t$ iterations.
Obviously, $[C_0(0),\dotsc,C_k(0)]$ obeys the multinomial distribution
$\mathrm{Multinom}(m, p_0(n), p_1(n),\dotsc, p_k(n))$ where
\begin{align*}
p_j(n) := \frac{\binom{n-u(n)}{j}\binom{u(n)}{k-j}}{\binom{n}{k}}=
\binom{k}{j}\frac{u(n)^{k-j}}{n^{k-j}} + O\left(\frac{u(n)^{k-j+1}}{n^{k-j+1}}\right).
\end{align*}
Let $[B_{1}(t),B_2(t),\dotsc, B_{k-r+1}(t)]$ be a 0-1 random vector of weight 1 where $B_j(t)=1$ if a constraint of the degree $j$ is chosen at $(t+1)$-th 
iteration
and $B_j(t)=0$ otherwise.
We assume that a constraint is chosen uniformly from all constraints of the degree at most $k-r+1$.
Hence, 
\begin{equation*}
\Pr(B_j(t)=1\mid [C_0(t),\dotsc,C_k(t)]) = \frac{C_j(t)}{\sum_{j'=1}^{k-r+1} C_{j'}(t)}
\end{equation*}
if $\sum_{j=1}^{k-r+1}C_j(t)\ge 1$.
Let $N(t):=n-u(n)-t$ be the number of remaining variable nodes after $t$ iterations when the iterations continues until $t$-th step.
The set of random variables $([C_0(t),\dotsc,C_k(t)])_{t=0,1,\dotsc,N(0)}$
is a Markov chain satisfying $[C_0(t+1),\dotsc,C_k(t+1)]=[C_0(t),\dotsc,C_k(t)]$ if $\sum_{j=1}^{k-r+1}C_j(t)=0$ and
\begin{equation}
\begin{split}
C_k(t+1)&= C_k(t) - R_k(t)\\
C_{j}(t+1)&= C_{j}(t) - R_{j}(t) + R_{j+1}(t),\hspace{2em} \text{for}\hspace{1em} j=1,2,\dotsc,k-1\\
C_0(t+1)&= C_0(t) + R_1(t)
\end{split}
\label{eq:Markov}
\end{equation}
if $\sum_{j=1}^{k-r+1}C_j(t)\ge 1$
where $R_1(t),\dotsc, R_k(t)$ are independent random variables conditioned on $[C_0(t),\dotsc,C_k(t)]$ and $[B_1(t),\dotsc,B_{k-r+1}(t)]$ 
obeying
\begin{align*}
R_j(t) &\sim \mathrm{Binom}\left(C_j(t), \frac{j}{N(t)}\right),\hspace{2em} \text{for}\hspace{1em} j=k-r+2, k-r+3,\dotsc,k\\
R_j(t) &\sim B_j(t)+\mathrm{Binom}\left(C_j(t)-B_j(t), \frac{j}{N(t)}\right),\hspace{2em} \text{for}\hspace{1em} j=1,2,\dotsc,k-r+1.
\end{align*}
Similar analysis was considered in~\cite{Luby:1997:PLC:258533.258573},  \cite{Achlioptas2001159}, \cite{connamacher2012exact}.
In the papers, the number of constraints $m$ is proportional to $n$.
In that case, one can use Wormald's theorem which gives differential equations describing the behavior of the Markov chain.
In this paper, $m$ is not necessarily proportional to $n$.
Hence, different techniques are required.
Let $E_1^{k-r+1}(t):= \sum_{j=1}^{k-r+1} j C_j(t)$ be the number of edges connected to constraints of the degree at most $k-r+1$.
Then, the probability that the randomly generated bipartite graph does not have $(k-r+2)$-stopping set of size larger than $n-u(n)-t$ is
\begin{equation}
\Pr\left(E_1^{k-r+1}(0)\ge 1,E_1^{k-r+1}(1)\ge 1,\dotsc,E_1^{k-r+1}(t-1)\ge 1\right).
\label{eq:ps}
\end{equation}
For proving Lemma~\ref{lem:2ML}, Theorems~\ref{thm:ML} and \ref{thm:MU}, we analyze the probability~\eqref{eq:ps}.

\subsection{Markov chain for upper bound}
In this subsection, we show Lemma~\ref{lem:2ML} and Theorem~\ref{thm:ML}.
For the Markov chain~\eqref{eq:Markov}, it holds
\begin{equation}
\begin{split}
C_k(t+1)&= C_k(t) - R_k(t)\\
C_{j}(t+1)&= C_{j}(t) - R_{j}(t) + R_{j+1}(t),\hspace{2em} \text{for}\hspace{1em} j=k-r+2,k-r+3,\dotsc,k-1\\
E_1^{k-r+1}(t+1)&= E_1^{k-r+1}(t) - \sum_{j=1}^{k-r+1}R_j(t) + (k-r+1) R_{k-r+2}(t)
\end{split}
\label{eq:MarkovE}
\end{equation}
if $E_1^{k-r+1}(t)\ge 1$.
For upper bounding~\eqref{eq:ps}, we consider a new Markov chain
$([\overline{E}_1^{k-r+1}(t),\overline{C}_{k-r+2}(t),\dotsc,\overline{C}_k(t)])_{t=0,1,\dotsc,N(0)}$
 which satisfies 
$\underline{E}_1^{k-r+1}(0)=\sum_{j=1}^{k-r+1}jC_j(0)$,
$\overline{C}_j(0)=C_j(0)$ for $j=k-r+2,\dotsc,k$ and
\begin{equation}
\begin{split}
\overline{C}_k(t+1)&= \overline{C}_k(t)\\
\overline{C}_{j}(t+1)&= \overline{C}_{j}(t) + \overline{R}_{j+1}(t),\hspace{2em} \text{for}\hspace{1em} j=k-r+2,k-r+3,\dotsc,k-1\\
\overline{E}_1^{k-r+1}(t+1)&= \overline{E}_1^{k-r+1}(t) - 1 + (k-r+1)\overline{R}_{k-r+2}(t)\\
\end{split}
\label{eq:MarkovU}
\end{equation}
where
\begin{align*}
\overline{R}_j(t) &\sim \mathrm{Binom}\left(\overline{C}_j(t), \frac{j}{N(t)}\right),\hspace{2em} \text{for}\hspace{1em} j=k-r+2,k-r+3,\dotsc,k.
\end{align*}
The new Markov chain does not include the condition $E_1^{k-r+1}(t)\ge 1$ which appears in~\eqref{eq:MarkovE}.
Hence, it is easier to analyze the new Markov chain than the original Markov chain~\eqref{eq:Markov}.
Obviously,~\eqref{eq:ps} is upper bounded by
\begin{equation}
\Pr\left(\overline{E}_1^{k-r+1}(0)\ge 1,\overline{E}_1^{k-r+1}(1)\ge 1,\dotsc,\overline{E}_1^{k-r+1}(t-1)\ge 1\right).
\label{eq:psU}
\end{equation}
The following theorem is proved in Section~\ref{sec:evo}.

\begin{theorem}[Moment generating function of $\overline{E}_1^{k-r+1}(t)$]
\label{thm:UG}
Assume $m=\mu\frac{n^{r-1}}{u(n)^{r-2}}$ for arbitrary constant $\mu$ and $u(n)\in \omega(1)\cap o(n)$.
Then, for any constants $\tau>0$ and $\lambda$, it holds
$\mathbb{E}[\exp\{\lambda\overline{E}_1^{k-r+1}(\tau u(n))\}]=
\exp\{ \varphi_{k,r}(\mu,\lambda,\tau) u(n)+ O(\max\{1,u(n)^2/n\})\}$
where $\varphi_{k,r}(\mu,\lambda,\tau)$ is defined in~\eqref{eq:Ek}.
\end{theorem}
Lemma~\ref{lem:2ML} and Theorem~\ref{thm:ML} can be proved by using Theorem~\ref{thm:UG} and the Chernoff bound.

\begin{proof}[Proof of Lemma~\ref{lem:2ML} and Theorem~\ref{thm:ML}]
From the Chernoff bound and Theorem~\ref{thm:UG}, one obtains an inequality
\begin{align*}
\Pr\left(\overline{E}_1^{k-r+1}(\tau u(n)) \ge 1\right)
&\le\Pr\left(\overline{E}_1^{k-r+1}(\tau u(n)) \ge 0\right)\\
&\le\mathbb{E}[\exp\{\lambda\overline{E}_1^{k-r+1}(\tau u(n))\}]
=\exp\{ \varphi_{k,r}(\mu,\lambda,\tau) u(n)+ O(\max\{1,u(n)^2/n\})\}
\end{align*}
for any constants $\tau\ge 0$ and $\lambda\ge 0$.
It holds
\begin{align*}
\frac{\partial \varphi_{k,r}(\mu,\lambda,\tau)}{\partial\lambda}
&=
\mu \exp\{(k-r+1)\lambda\}(k-r+1)\binom{k}{r-1}(1+\tau)^{r-1}- \tau\\
&=\mu \exp\{(k-r+1)\lambda\}r\binom{k}{r}(1+\tau)^{r-1}- \tau.
\end{align*}
If
\begin{equation}\label{eq:cond}
\left.\frac{\partial \varphi_{k,r}(\mu,\lambda,\tau)}{\partial\lambda}\right|_{\lambda=0}=
\mu r\binom{k}{r}(1+\tau)^{r-1}- \tau < 0
\end{equation}
then $\varphi_{k,r}(\mu,\lambda,\tau)$ is negative for sufficiently small $\lambda>0$ since $\varphi_{k,r}(\mu,0,\tau)=0$.
The condition~\eqref{eq:cond} is satisfied for some $\tau>0$ when
\begin{equation}\label{eq:sup}
\mu
< \frac1{r\binom{k}{r}}\sup_{\tau> 0} \frac{\tau}{(1+\tau)^{r-1}}.
\end{equation}
When $r=2$ the supremum is taken at $\tau\to+\infty$, and hence the condition~\eqref{eq:sup} is equivalent to $\mu < [k(k-1)]^{-1}$.
When $r\ge 3$ the supremum is taken at $\tau=1/(r-2)$, and hence the condition~\eqref{eq:sup} is equivalent to $\mu < \frac{(r-2)^{r-2}}{\binom{k}{r}r(r-1)^{r-1}}$.
If the condition~\eqref{eq:sup} is satisfied, then~\eqref{eq:cond} is satisfied for some $\tau$.
That means that there exists $(k-r+2)$-stopping set of size at least $n-(1+\tau)u(n)$.
By optimizing the Chernoff bound, one obtains Theorem~\ref{thm:ML}.
\end{proof}

\subsection{Markov chain for lower bound}
In this subsection, we prove Theorem~\ref{thm:MU}.
We can use the same argument as Lemma~\ref{lem:2large} for the $(k-r+2)$-peeling algorithm.
For $m=\mu\frac{n^{r-1}}{u(n)^{r-2}}$, it holds
\begin{equation*}
\mathbb{E}[C_{k-r+2}(0)]
=m p_{k-r+2}(n) = \mu \binom{k}{r-2} n + O(u(n)).
\end{equation*}
Hence, if $\mu > [r(r-1)\binom{k}{r}]^{-1}$,
it holds
$\mathbb{E}[C_{k-r+2}(0)]>([(k-r+2)(k-r+1)]^{-1} + \delta)n$ for sufficiently small $\delta>0$.
In this case, from the argument in the proof of Lemma~\ref{lem:2large}, linearly many variables are removed
 by the $(k-r+2)$-peeling algorithm with high probability.
However, $[r(r-1)\binom{k}{r}]^{-1}$ is strictly larger than $\mu_\mathrm{c}(k,r)$ for $r\ge 3$.
In the following, we will show that if $\mu>\mu_\mathrm{c}(k,r)$,
for any $\eta>0$ there exists $\tau>0$ such that
\begin{equation}\label{eq:psl}
\Pr\left(E_1^{k-r+1}(0)\ge 1,\dotsc, E_1^{k-r+1}(\tau u(n)-1)\ge 1, E_1^{k-r+1}(\tau u(n))\ge \eta u(n)\right)=1-o(1)
\end{equation}
and that if $\mu>\mu_\mathrm{c}(k,r)$,
there exists sufficiently small $\epsilon >0$ such that for any $\tau\ge 1/(r-2)$,
\begin{equation}\label{eq:p2l}
\Pr\left(C_{k-r+2}(\tau u(n))>[(k-r+2)(k-r+1)+\epsilon]^{-1} n\right) = 1-o(1).
\end{equation}
They give the proof of Theorem~\ref{thm:MU} except for the bound of probability.

For lower bounding the probabilities in~\eqref{eq:psl} and~\eqref{eq:p2l}, we consider a new Markov chain
$([\underline{E}_1^{k-r+1}(t),\underline{C}_{k-r+2}(t),\dotsc,\underline{C}_k(t)])_{t=0,1,\dotsc,N(0)}$
 which satisfies 
$\underline{E}_1^{k-r+1}(0)=\sum_{j=1}^{k-r+1}jC_j(0)$,
$\underline{C}_j(0)=C_j(0)$ for $j=k-r+2,\dotsc,k$ and
\begin{equation}
\begin{split}
\underline{C}_k(t+1)&= \underline{C}_k(t)-\underline{R}_k(t)\\
\underline{C}_{j}(t+1)&= \underline{C}_{j}(t)-\underline{R}_{j+1}(t) + \underline{R}_{j+1}(t),
\hspace{2em} \text{for}\hspace{1em} j=k-r+2,k-r+3,\dotsc,k-1\\
\underline{E}_1^{k-r+1}(t+1)&= \underline{E}_1^{k-r+1}(t) - 1 -\underline{R}_1^{k-r+1}(t) + (k-r+1)\underline{R}_{k-r+2}(t)\\
\end{split}
\label{eq:MarkovL}
\end{equation}
where
\begin{align*}
\underline{R}_j(t) &\sim \mathrm{Binom}\left(\underline{C}_j(t), \frac{j}{N(t)}\right),\hspace{2em} \text{for}\hspace{1em} j=k-r+2,k-r+3,\dotsc,k\\
\underline{R}_1^{k-r+1}(t) &\sim \mathrm{Binom}\left(\underline{E}_1^{k-r+1}(t)+t, \frac{1}{N(t)-k+r}\right).
\end{align*}
We obtain a lower bound of the probabilities in~\eqref{eq:psl} and \eqref{eq:p2l} by replacing the original Markov chain by the above new Markov chain.

\begin{theorem}[Moment generating function of $\underline{E}_1^{k-r+1}(t)$]
\label{thm:LG}
Assume $m=\mu\frac{n^{r-1}}{u(n)^{r-2}}$ for arbitrary constant $\mu$ and $u(n)\in \omega(1)\cap o(n)$.
Then, for any constants $\tau>0$ and $\lambda$, it holds
$\mathbb{E}[\exp\{\lambda\underline{E}_1^{k-r+1}(\tau u(n))\}]=
\exp\{ \varphi_{k,r}(\mu,\lambda,\tau) u(n)+ O(\max\{1,u(n)^2/n\})\}$
where $\varphi_{k,r}(\mu,\lambda,\tau)$ is defined in~\eqref{eq:Ek}.
\end{theorem}
The proof is omitted since it is straightforward from the proof of Theorem~\ref{thm:UG}.
From Theorem~\ref{thm:LG}, if $\mu>\mu_\mathrm{c}(k,r)$, 
it holds
\begin{align}
\Pr\left(\bigcup_{t=0}^{\tau u(n)-1}\underline{E}_1^{k-r+1}(t)\le 0\right)&\le
\sum_{t=0}^{\tau u(n)-1}\Pr\left(\underline{E}_1^{k-r+1}(t)\le 0\right)\nonumber\\
&\le \sum_{t=0}^{\tau u(n)-1} \inf_{\lambda<0}\mathbb{E}\left[\exp\left\{\lambda\underline{E}_1^{k-r+1}(t)\right\}\right]\nonumber\\
&\le \exp\left\{\sup_{\tau'>0}\inf_{\lambda<0}\varphi_{k,r}(\mu,\lambda,\tau')u(n)+ O(\max\{u(n)^2/n,\log u(n)\})\right\}.
\label{eq:supinf}
\end{align}
Note that the above upper bound is independent of $\tau$.
In the same way, one can show that if $\mu<\mu_\mathrm{c}(k,r)$, for sufficiently large $\tau>0$, it holds
\begin{align*}
&\Pr\left(\underline{E}_1^{k-r+1}(\tau u(n))\le \eta u(n)-1\right)\le \exp\{-c_{\tau,\eta}u(n)\}
\end{align*}
for some constant $c_{\tau,\eta}>0$ depending on $\tau$ and $\eta$ which tends to infinity as $\tau\to\infty$ while $\eta$ is fixed.

Similarly to Theorem~\ref{thm:LG}, asymptotic analysis of the moment generating function for $\underline{C}_{k-r+2}(t)$ is obtained for $t=O(u(n))$. 
\begin{theorem}[Moment generating function of $\underline{C}_j(t)$]\label{thm:LG2}
Assume $m=\mu\frac{n^{r-1}}{u(n)^{r-2}}$ for arbitrary constant $\mu$ and $u(n)\in \omega(1)\cap o(n)$.
Then, for any constants $\tau>0$ and $\lambda_j$,
$\mathbb{E}[\exp\{\lambda_j\underline{C}_j(\tau u(n))\}]=
\exp\{ \varphi^{(j)}_{k,r}(\mu,\lambda_j,\tau) \frac{n^{j-k+r-1}}{u(n)^{j-k-r-2}}+ O(\frac{u(n)^{j-k+r-1}}{n^{j-k+r-1}}\max\{1,u(n)^2/n\})\}$
where
\begin{equation*}
\varphi^{(j)}_{k,r}(\mu,\lambda,\tau):= \mu \left(\exp\{\lambda\}-1\right)\binom{k}{k-j}(1+\tau)^{k-j}.
\end{equation*}
\end{theorem}
The proof of this theorem is also omitted since it is straightforward from the proof of Theorem~\ref{thm:UG}.
From Theorem~\ref{thm:LG2}, it holds
\begin{align*}
&
\Pr\left(\underline{C}_{k-r+2}(\tau u(n))\le([(k-r+2)(k-r+1)]^{-1}+\epsilon)n\right)
\le
\frac{\mathbb{E}\left[\exp\left\{\lambda_{k-r+2} \underline{C}_{k-r+2}(\tau u(n))\right\}\right]}{\exp\{\lambda_{k-r+2}([(k-r+2)(k-r+1)]^{-1}+\epsilon)n\}}\\
&\le
\frac{\mathbb{E}\left[\exp\left\{\lambda_{k-r+2} \underline{C}_{k-r+2}(u(n)/(r-2))\right\}\right]}{\exp\{\lambda_{k-r+2}([(k-r+2)(k-r+1)]^{-1}+\epsilon)n\}}\\
&=
\exp\left\{\mu \left(\exp\{\lambda_{k-r+2}\}-1\right)\frac{1}{(k-r+2)(k-r+1)}\binom{k}{r}\frac{r(r-1)^{r-1}}{(r-2)^{r-2}} n - \lambda_{k-r+2}([(k-r+2)(k-r+1)]^{-1}+\epsilon)n\right\}\\
&=
\exp\left\{[(k-r+2)(k-r+1)]^{-1}\frac{\mu}{\mu_\mathrm{c}(k,r)} \left(\exp\{\lambda_{k-r+2}\}-1\right) n - \lambda_{k-r+2}([(k-r+2)(k-r+1)]^{-1}+\epsilon)n\right\}
\end{align*}
for any $\lambda_{k-r+2}\le 0$.
Hence, if $\mu>\mu_\mathrm{c}(k,r)$, for sufficiently small $\epsilon>0$, there is $\delta>0$ such that
\begin{align*}
\Pr\left(\underline{C}_{k-r+2}(u(n)/(r-2))\le([(k-r+2)(k-r+1)]^{-1}+\epsilon)n\right)&\le
\exp\{-\delta n\}.
\end{align*}
From the argument in the proof of Lemma~\ref{lem:2large}, a subgraph of the bipartite graph at the $\tau u(n)$-th step
including all variable vertices and all constraint vertices of the degree $k-r+2$ has a giant component with probability exponentially close to 1
with respect to $n$.
The number of variables which can be removed is $E_1^{k-r+1}(\tau u(n))$ which is larger than $\eta u(n)$ with probability
 at least $1-\exp\{-c_{\tau,\eta}u(n)\}$ for sufficiently large $c_{\tau,\eta}$ when sufficiently large $\tau>0$ is chosen.
In that case, the $(k-r+2)$-peeling algorithm removes linearly many variables with probability $1-\rho^{\eta u(n)}$ where the size of 
the giant component is $(1-\rho)N(\tau u(n))$ .
If one chooses sufficiently large $\eta>0$, the probability that the $(k-r+2)$-peeling algorithm fails to remove linearly many variable
is dominated by~\eqref{eq:supinf}.

\section{Evolution of the moment generating function}\label{sec:evo}
In this section, the proof of Theorem~\ref{thm:UG} is shown.
The moment generating function for $[(\overline{E}_1^{k-r+1}(t)+t)/(k-r+1), \overline{C}_{k-r+2}(t),\dotsc,\overline{C}_k(t)]$ is defined as
\begin{align*}
\overline{f}_t(\lambda_{k-r+1},\dotsc,\lambda_k):=\mathbb{E}\left[
\exp\left\{\lambda_{k-r+1}(\overline{E}_1^{k-r+1}(t)+t)/(k-r+1)+\lambda_{k-r+2}\overline{C}_{k-r+2}(t)+\dotsb+\lambda_k\overline{C}_k(t)\right\}\right].
\end{align*}
From~\eqref{eq:MarkovU}, one obtains a recursive formula
\begin{align*}
\overline{f}_{t+1}(\lambda_{k-r+1},\dotsc,\lambda_k)&=
\mathbb{E}\left[\exp\left\{\lambda_{k-r+1}(\overline{E}_1^{k-r+1}(t+1)+t+1)/(k-r+1)+\lambda_{k-r+2}\overline{C}_{k-r+2}(t)+\dotsb+\lambda_k\overline{C}_k(t+1)\right\}\right]\\
&=\mathbb{E}\Biggl[\exp\left\{\lambda_{k-r+1}(\overline{E}_1^{k-r+1}(t)+t)/(k-r+1)+\lambda_{k-r+2}\overline{C}_{k-r+2}(t)+\dotsb+\lambda_k\overline{C}_k(t)\right\}\\
&\qquad\cdot
\exp\left\{\lambda_{k-r+1}\overline{R}_{k-r+2} + \lambda_{k-r+2}\overline{R}_{k-r+3}+\dotsb+\lambda_{k-1}\overline{R}_k\right\}\Biggr]\\
&=\mathbb{E}\Biggl[\exp\left\{\lambda_{k-r+1}(\overline{E}_1^{k-r+1}(t)+t)/(k-r+1)+\lambda_{k-r+2}\overline{C}_{k-r+2}(t)+\dotsb+\lambda_k\overline{C}_k(t)\right\}\\
&\quad\cdot\prod_{j=k-r+2}^{k}\left(1-\frac{j}{N(t)}+\frac{j}{N(t)}\exp\{\lambda_{j-1}\}\right)^{\overline{C}_j(t)}\Biggr]\\
&= \overline{f}_t(\lambda_{k-r+1},\lambda'_{k-r+2},\dotsc,\lambda'_k)
\end{align*}
where
\begin{align*}
\lambda'_j &:= \lambda_j + \log\left(1-\frac{j}{N(t)}+\frac{j}{N(t)}\exp\{\lambda_{j-1}\}\right)
\end{align*}
for $j=k-r+2,k-r+3,\dotsc,k$.
Let $\lambda_{k-r+1}^{(s)}:= \lambda_{k-r+1}$ for $s=1,2,\dotsc,t$.
For $j=k-r+2,k-r+3,\dotsc,k$,
$\lambda^{(0)}_j:= 0$ and
\begin{align*}
\lambda^{(s)}_j &:= \lambda^{(s-1)}_j + \log\left(1-\frac{j}{N(t-s+1)}+\frac{j}{N(t-s+1)}\exp\{\lambda^{(s-1)}_{j-1}\}\right),
\end{align*}
for $s=1,2,\dotsc,t$.
Then, it holds
\begin{align*}
\mathbb{E}[\exp\{\lambda_{k-r+1} (\overline{E}_1^{k-r+1}(t)+t)/(k-r+1)\}]
=\overline{f}_{t}(\lambda_{k-r+1},0,\dotsc,0)
&= \overline{f}_0(\lambda_{k-r+1}^{(t)},\lambda_{k-r+2}^{(t)},\dotsc,\lambda_k^{(t)}).
\end{align*}

\begin{lemma}
For $t=O(u(n))$ and $u(n)=o(n)$, it holds
\begin{align*}
\exp\{\lambda_j^{(t)}\} &= 
 1 + \binom{j}{k-r+1}\frac{t^{j-k+r-1}}{n^{j-k+r-1}}\left(\exp\{\lambda_{k-r+1}\}-1\right)
 + O\left(\frac{u(n)^{j-k+r-2}}{n^{j-k+r-1}}\max\left\{1,\frac{u(n)^2}{n}\right\}\right)
\end{align*}
for $j=k-r+1,k-r+2,\dotsc,k$.
\begin{proof}
The lemma is shown by induction on $j$. 
The lemma obviously holds for $j=k-r+1$.
Assume the lemma holds for $j=j_0-1\ge 1$, then
\begin{align*}
\lambda_{j_0}^{(t)} &= \sum_{s=0}^{t-1}\log\left(1-\frac{j_0}{N(t-s)}+\frac{j_0}{N(t-s)}\exp\{\lambda_{j_0-1}^{(s)}\}\right)\\
&= \sum_{s=0}^{t-1}\frac{j_0}{N(t-s)}\left(\exp\{\lambda_{j_0-1}^{(s)}\}-1\right) + O\left(\frac{u(n)^{2(j_0-k+r-1)-1}}{n^{2(j_0-k+r-1)}}\right)\\
&= \sum_{s=0}^{t-1}\frac{j_0}{n}\left(\exp\{\lambda_{j_0-1}^{(s)}\}-1\right) + O\left(\frac{u(n)^{j_0-k+r}}{n^{j_0-k+r}}\right)\\
&= \sum_{s=0}^{t-1}\frac{j_0}{n}\binom{j_0-1}{k-r+1}\frac{s^{j_0-k+r-2}}{n^{j_0-k+r-2}}(\exp\{\lambda_{k-r+1}\}-1) + O\left(\frac{u(n)^{j_0-k+r-2}}{n^{j_0-k+r-1}}\max\left\{1,\frac{u(n)^2}{n}\right\}\right)\\
&= \binom{j_0}{k-r+1}\frac{t^{j_0-k+r-1}}{n^{j_0-k+r-1}}(\exp\{\lambda_{k-r+1}\}-1) + O\left(\frac{u(n)^{j_0-k+r-2}}{n^{j_0-k+r-1}}\max\left\{1,\frac{u(n)^2}{n}\right\}\right).
\qedhere
\end{align*}
\end{proof}
\end{lemma}

Since $[C_0(0),\dotsc,C_k(0)]$ obeys the multinomial distribution $\mathrm{Multinom}(m,p_0(n),\dotsc,p_k(n))$,
it holds for $t=\tau u(n)$ and $m=\mu\frac{n^{r-1}}{u(n)^{r-2}}$ that
\begin{align*}
&\overline{f}_0(\lambda_{k-r+1}^{(t)},\lambda_{k-r+2}^{(t)},\dotsc,\lambda_k^{(t)})
=\left(p_0(n) + \sum_{j=1}^{k-r+1}p_j(n)\exp\left\{\frac{j}{k-r+1}\lambda_{k-r+1}\right\}
+ \sum_{j=k-r+2}^k p_j(n)\exp\left\{\lambda_j^{(t)}\right\}\right)^m\\
&=\left(1+\sum_{j=k-r+1}^k p_j(n)\binom{j}{k-r+1}\frac{t^{j-k+r-1}}{n^{j-k+r-1}}(\exp\{\lambda_{k-r+1}\}-1)
 + O\left(\frac{u(n)^{r-2}}{n^{r-1}}\max\left\{1,\frac{u(n)^2}{n}\right\}\right)\right)^m\\
&=\left(1+\frac{u(n)^{r-1}}{n^{r-1}} \left(\exp\{\lambda_{k-r+1}\}-1\right)\sum_{j=k-r+1}^k \binom{k}{j}\binom{j}{k-r+1}\tau^{j-k+r-1}
 + O\left(\frac{u(n)^{r-2}}{n^{r-1}}\max\left\{1,\frac{u(n)^2}{n}\right\}\right)\right)^m\\
&=\exp\left\{u(n)\mu \left(\exp\{\lambda_{k-r+1}\}-1\right)\binom{k}{r-1}(1+\tau)^{r-1}+ O\left(\max\left\{1,\frac{u(n)^2}{n}\right\}\right)\right\}.
\end{align*}
From
\begin{align*}
\mathbb{E}[\exp\{\lambda (\overline{E}_1^{k-r+1}(\tau u(n))+\tau u(n))/(k-r+1)\}]
&=\exp\left\{u(n)\mu \left(\exp\{\lambda\}-1\right)\binom{k}{r-1}(1+\tau)^{r-1}+ O\left(\max\left\{1,\frac{u(n)^2}{n}\right\}\right)\right\}
\end{align*}
one obtains
\begin{align*}
\mathbb{E}[\exp\{\lambda \overline{E}_1^{k-r+1}(\tau u(n))\}]
&=\exp\left\{u(n)\left[\mu \left(\exp\{(k-r+1)\lambda\}-1\right)\binom{k}{r-1}(1+\tau)^{r-1}-\lambda\tau\right]+ O\left(\max\left\{1,\frac{u(n)^2}{n}\right\}\right)\right\}.
\end{align*}

\section{Conclusion and discussion}\label{sec:condis}
In this paper, the tight thresholds for small and large stopping sets on randomly generated bipartite graph are shown.
This result gives the tight threshold for the LP relaxation attack for Goldreich's generator using $(r-1)$-wise independent local functions.
When $u(n)=O(\log n)$, the order $\frac{n^{r-1}}{u(n)^{r-2}}$ is much larger than $n^{r/2}$.
Hence, the LP using the local marginal polytope is suboptimal.
This problem can be avoided by using a tighter polytope
\begin{equation*}
\mathrm{CVM}:=\Bigl\{(q_C(\bm{x}_C))_{C\in \mathcal{C}_F, x_C\in\{0,1\}^{|C|}}\mid
 q_C\in \mathrm{DIST}_{|C|},\,\forall C\in \mathcal{C}_F, q_{C'}(\bm{x}_{C'})=\sum_{\bm{x}_{C\setminus C'}} q_C(\bm{x}_C),\,\forall C'\subseteq C\Bigr\}
\end{equation*}
where $\mathcal{C}_F\subseteq 2^{[n]}$ is a set includes $\partial (a)$ and its subsets for all $a\in[m]$.
One can easily find an example of $(r-1)$-wise independent function for which the LP relaxation using $\mathrm{CVM}$
can only have integral zero-optimal solutions when $m=C n^{r/2}$ for sufficiently large constant $C$, e.g., $r$-XORSAT.
Note that the above polytope corresponds to the cluster variation method in statistical physics~\cite{yedidia2005constructing}.
One can also consider the generalized BP and its linearization corresponding to the above polytope~\cite{yedidia2005constructing}.
The derivation of the threshold constant for the stronger LP is an interesting problem.

\bibliographystyle{IEEEtran}
\bibliography{IEEEabrv,bibliography}

\end{document}